\documentclass[12pt,leqno]{article}
\usepackage{amsmath,amstext, amsthm,amsfonts,amssymb,amsthm} 
\usepackage{multicol}

\usepackage{latexsym}
\usepackage{color}
\usepackage{graphicx}
\usepackage{epsf}
\usepackage{epsfig}
\usepackage{subfigure}
\usepackage{rotating}
\usepackage{curves}
\usepackage{epic}

\usepackage{graphics}
\usepackage{verbatim}
\usepackage{color}
\usepackage{hyperref}
 \numberwithin{equation}{section}
\def\R{\mathbb{ R}}  
  
\usepackage{epsfig}
\newtheorem{thm}{Theorem}[section]

\renewcommand{\O}{\mathcal{O}}

\newtheorem{rem}[thm]{Remark}
\newcommand{\be}{\begin{equation}}
\newcommand{\ee}{\end{equation}}

\newcommand{\ba}{\begin{array}}
\newcommand{\ea}{\end{array}}

\renewcommand{\em}{\it}
\newcommand{\al}{\alpha}
\newcommand{\bt}{\beta}
\newcommand{\s}{\sigma}

\renewcommand{\O}{\mathcal{O}}
\newcommand{\la}{\lambda}
\newcommand{\f}{\phi}

\newcommand{\G}{\Gamma}

\newcommand{\bea}{\begin{eqnarray}}
\newcommand{\eea}{\end{eqnarray}}

\newcommand{\Sum}{\sum_{n=0}^\infty}

\newcommand{\h}{\mathfrak H}

\title{Some Orthogonal Polynomials Arising from\\
       Coherent States}
    \author{S. Twareque Ali \\
    Department of Mathematics and Statistics\\
    Concordia University \\
    Montr\'eal, Qu\'ebec \\ CANADA H3G 1M8
\and
Mourad E. H. Ismail
    \\ City University of Hong Kong,\\ Tat Chee Avenue, Kowloon, 
    Hong Kong \\and King Saud University, Riyadh, Saudi Arabia}
\begin{document}
 \maketitle
\begin{abstract}
We explore in this paper some orthogonal polynomials which are naturally associated to certain families of coherent states, often referred to as nonlinear coherent states in the quantum optics literature.  Some examples turn out to be known orthogonal polynomials but in many cases we encounter a general class of new orthogonal polynomials for which  we establish several qualitative results.
 \end{abstract}


\medskip


\medskip

{\bf AMS} subject classification: 46N50,  81Q99, 12Y99

\section{Introduction}\label{Introduction}\label{sec-intro}
Coherent states are well known objects, both in physics and  mathematics (see, for example, \cite{AAG-book} and references cited therein). Their use in physics goes back to the early days of quantum mechanics, starting with the 1926 paper by Schr\"odinger \cite{Schroed} and their rediscovery in the theory of optical coherence several decades later \cite{glauber,klauder-sudarsh}. In recent years coherent states have been extensively used in many areas of physics, e.g., quantum optics, atomic and molecular physics, etc.,  (see, for example, \cite{gazeau-book} for a detailed account). Their mathematical properties, specially in the context of quantization theory, square-integrable group representation theory, symplectic geometry, etc., have also been extensively studied in recent decades.
In this paper we look at two sets of orthogonal polynomials that are naturally associated to a particular family of coherent states, known in the quantum optical literature as {\em nonlinear coherent states} (see \cite{AAG-book} and \cite{sivakum} for an extensive discussion). One set of these polynomials, arising from the {\em shift operators} associated to these coherent states, have been known and studied before \cite{borzov,borzov2,odz,odhorter}. Here we take another look at these polynomials, but study a different aspect of their structure. We also show how a second set of orthogonal polynomials can be obtained from the measure which gives the {\em resolution of the identity} for the coherent states. The moments of this measure are obtained from $n$-term partial products of the terms of the sequence defining the shift operators. In this way, both sets of orthogonal polynomials are intimately related. While it would be our aim to eventually take a closer look at the relationship between these two sets of polynomials, in the present paper we mainly focus on  measures of the type which arise from the above mentioned resolution of the identity and for them discuss the associated families of orthogonal polynomials.

We start out by introducing the notion of   nonlinear coherent states mathematically and the derivation of their associated  polynomials, in the framework within which we would like to study them here.

\subsection{Nonlinear coherent states}\label{subsec-coh-st}
We ought to mention at the outset that the term nonlinear, as applied to coherent states, does not refer to any mathematical nonlinearity, but rather is a reflection of their appearance in nonlinear optics.
Generally, a family of such coherent states is an {\em overcomplete} set of vectors in a  Hilbert space, labelled by a continuous parameter $z$ which runs over a complex domain. The vectors are, in addition, subject to a {\em resolution of the identity condition\/.} More precisely, let $\h$ be a (complex, separable, infinite dimensional) Hilbert space, $\{\phi_n\}_{n=0}^\infty$ an orthonormal basis of it and let $\{x_n\}_{n=0}^\infty, \; x_0 = 0$, be an infinite sequence of positive numbers. Let $\lim_{n\rightarrow \infty}x_n = L^2$, where $L > 0$ could be finite or infinite, but not zero. We shall use the notation $x_n! = x_1 x_2 \ldots x_n$ and  $x_0! = 1$. For each $z\in \mathcal D$ (some domain in $\mathbb C$), we define a non-linear coherent state, i.e., a vector $\eta_z \in \h$, in the manner
 \begin{equation}
  \eta_z = \mathcal N(\vert z\vert^2 )^{-\frac 12} \sum_{n=0}^\infty\frac {z^n}{\sqrt{x_n!}}\phi_n\; ,
\label{CS-def}
\end{equation}
where the normalization constant $\mathcal N(\vert z\vert^2 ) = \sum_{n=0}^\infty\dfrac {\vert z\vert^{2n}}{x_n!}$ is chosen so that $\Vert \eta_z \Vert = 1$. It is clear that the vectors $\eta_z$ are well defined for all $z$ for which the above sum, representing $\mathcal N(\vert z\vert^2 )$, converges, i.e., $\mathcal D = \{ z\in \mathbb C \mid \vert z \vert < L\}$. Furthermore, we require that there exist a measure $d\nu (z, \overline{z})$ on $\mathcal D$ for which the resolution of the identity condition,
\begin{equation}
 \int_{\mathcal D} \vert \eta_z\rangle\langle \eta_z \vert\;\mathcal N(\vert z\vert^2 )\; d\nu (z, \overline{z}) =
     I_\h\; ,
\label{resolid}
\end{equation}
holds.

It is easily seen that in order for (\ref{resolid}) to be satisfied, $d\nu$ has to have the form,
\be
  d\nu (z , \overline{z}) = \frac 1{2\pi}\; d\theta\; d\lambda(r)\; , \qquad z = re^{i\theta}\; ,
\label{orthog-meas}
\end{equation}
where the measure $d\lambda$ is a solution of the moment problem,
\be
  \int_0^L r^{2n} d\lambda(r) = x_n! , \qquad n =0,1,2, \ldots  \; ,
\label{mom-prob}
\end{equation}
provided such a solution exists.  In most of the cases that occur in practice, the support of the measure $d\nu$ is the whole of $\mathcal D$, i.e., $d\lambda$ is supported on the entire interval $(0, L)$.

   Below are some examples of the above general construction.

\subsubsection{Canonical coherent states}\label{subsubsec-CCS}
Let $x_n = n$ so that $L = \infty$. In that case the coherent states $\eta_z$, defined for all $z\in \mathbb C$,  are the well-known  {\em canonical coherent states\/.}
\be
   \eta_z = e^{-\frac 12 \vert z \vert^2}\sum_{n=0}^\infty \frac {z^n}{\sqrt{n!}}\phi_n\; .
\label{CCS}
\end{equation}
The moment problem becomes,
\be
  \int_0^\infty r^{2n} d\lambda(r) = n! , \qquad n =0,1,2, \ldots  \; ,
\label{CCS-mom-prob}
\end{equation}
so that
\be
       d\lambda (r) = 2re^{-r^2}\; dr, \qquad 0\leq r < \infty\; .
\label{CCS-meas}
\end{equation}

\subsubsection{SU(1,1) discrete series coherent states}
This time $x_n = \dfrac n{2j + n -1}$ so that $x_n! = \dfrac {n!}{(2j)_n}$, where we have used the shifted factorials $(a)_n = a(a+1)(a+2) \ldots (a+n-1)$, and $j = 1, 1/2, 2, 3/2, 3, \ldots,$ is a constant. Also, $L = 1$ and the associated coherent states, defined for all $z$ in the unit disc $\mathcal D = \vert z\vert < 1$, are
\be
   \eta_z = (1 - r^2)^j \sum_{n=0}^\infty \left[\frac {(2j)_n}{n!}\right]^{\frac 12}z^n \phi_n\; ,
   \qquad r = \vert z \vert\; .
\label{sucs}
\end{equation}
The corresponding moment problem is now
\be
  \int_0^1 r^{2n} \; d\lambda (r) =  \frac {n!}{(2j)_n}\; ,
\label{su-mom-prob}
\end{equation}
which has the solution,
\be
  d\lambda (r) = 2(2j -1)\; r(1 - r^2)^{2j -2}\; dr, \qquad 0 \leq r < 1\; .
\label{su-meas}
\end{equation}
The coherent states, $\eta_z$ above, arise from the discrete series representations (paramet\-rized by $j$) of the $SU(1,1)$ group.

\subsubsection{Barut-Girardello coherent states}\label{subsubsec:bar-gir-cs}
There is a second set of coherent states associated to the $SU(1,1)$ group which are consructed using the {\em ladder operators} appearing in the Lie algebra of the group. These are known as the Barut-Girardello coherent states \cite{bar-gir}. For these coherent states, $X = \mathbb C, \; x_n = n(2j +n -1)$ and $x_n! = n! (2j)_n$ with $j = 1, 1/2, 2, 2/3 , \ldots$, as before. The coherent states are
\begin{equation}
  \eta_z = \frac {\vert z\vert^{2j -1}}{\sqrt{I_{2j -1}(2\vert z\vert )}}\;
     \sum_{n=0}^\infty\frac {z^n}{\sqrt{n! (2\kappa + n -1)!}}\; \phi_n \; , \qquad z \in \mathbb C,
\label{bar-gir-cs}
\end{equation}
where $I_\nu (x)$ denotes the order-$\nu$ modified Bessel function of the first kind. These
coherent states satisfy the resolution of the identity,
\begin{equation}
   \frac 2\pi\; \int_{\mathbb C}\mid \eta_z\rangle
   \langle \eta_z \mid\; K_{2j -1}(2r)\; I_{2j -1}
          (2r )\; r\; dr\; d\theta = I\; , \qquad z = r e^{i\theta}\; ,
\label{bar-gir-resolid}
\end{equation}
where again, $K_\nu (x)$ is the order-$\nu$ modified Bessel function of the second kind.

\subsubsection{Coherent states from analytic functions}\label{subsubsec-cs-anal-fcns}
  All three examples  above could be seen as special cases of a more general construction. Let $f(z)$ be an analytic function which has a Taylor expansion (around the origin) of the type,
\be
    f(z) = \sum_{n=0}^\infty\frac {z^n}{\rho (n)}\; , \qquad 0 < \rho (n) < \infty; , \quad
     \rho (0)= 1\; ,
\label{anal-fcn}
\end{equation}
and let $L = \lim_{n\rightarrow\infty}\dfrac {\rho(n+1)}{\rho(n)} >0$ be its radius of convergence. Then
$$ \mathcal N (\vert z \vert^2) = \sum_{n=0}^\infty\frac {\vert z\vert^{2n}}{[\rho (n)]^2} $$
converges in the disc $\vert z \vert < L$. Defining $x_n = \dfrac {\rho(n+1)}{\rho(n)}$, we can construct vectors $\eta_z$ in the Hilbert space $\h$ following
(\ref{CS-def}), which will be well defined in this disc and will constitute a family of coherent states provided the moment problem (\ref{mom-prob}) has a solution. It is clear that a large number of hypergeometric functions will lead to families of coherent states in this manner.

\subsection{Associated orthogonal polynomials}\label{subsec-orthog-polyn}
As mentioned earlier, there are two sets of orthogonal polynomials, naturally associated to a family of non-linear coherent states, that we now present.

\subsubsection{Polynomials associated to $d\lambda$}\label{subsubsec-assoc-polyn1}
The first of these sets is determined by the measure $d\lambda$. This measure can be extended to an even positive measure
\be
  d\mu (t) = \frac 12 d\lambda (\vert t\vert )\; ,
\label{ext-meas}
\end{equation}
on the symmetric interval $[-L , L]$, with moments,
\be
  \mu_{2n} = 2\int_0^\infty t^{2n}\; d\mu (t) = x_n! , \qquad \mu_{2n+1} = 2\int_0^\infty t^{2n+1}\; d\mu (t) = 0\; , n=0,1,2, \ldots\; .
\label{mu-mom}
\end{equation}
With these we can build a set of (monic) polynomials $P_n (x), \; n = 0,1,2, \ldots$, orthogonal with respect to the measure $d\mu$, in the usual manner \cite{Ism} using the Hankel determinant,
\be
  D_n = \begin{vmatrix} \mu_0 & \mu_1 & \ldots & \mu_n\\
                        \mu_1 & \mu_2 & \ldots & \mu_{n+1}\\
                        \vdots & \vdots & \ldots & \vdots \\
                        \mu_n & \mu_{n+1} & \ldots & \mu_{2n}\end{vmatrix},
  \quad P_n(x) = \frac 1{D_{n-1}} \begin{vmatrix} \mu_0 & \mu_1 & \ldots & \mu_n\\
                        \vdots & \vdots & \ldots & \vdots \\
                        \mu_{n-1} & \mu_n & \ldots & \mu_{2n-1}\\
                        1 & x & \ldots & x^n\end{vmatrix}.
\label{hank-det-pol}
\end{equation}

For example, for the canonical coherent states with the measure (\ref{CCS-meas}), these would be Laguerre polynomials in the variable $x^2$. For the $SU(1,1)$ discrete series case, with the measure (\ref{su-meas}), the associated polynomials are related to the Jacobi polynomials as explained  in Section \ref{subsec-pollaczek}, with $\alpha = \dfrac 12$ and $\beta = 2j -2$. For case of the Barut-Girardello coherent states the measure is
\be
 d\lambda (r) = \frac 2\pi K_{2j -1} (2r)\; r^{2-2j}\; dr\; .
\label{BG-meas}
\end{equation}
The corresponding polynomials $P_n$ can of course be computed using (\ref{hank-det-pol}), but are not among the well-known polynomials.

\subsubsection{Polynomials associated to shift operators}\label{subsubsec-assoc-polyn2}
To get the second set of polynomials, let us go back to the definition (\ref{CS-def}) and define the formal shift operators
\be
  a\phi_n = \sqrt{x_n}\phi_{n-1}, \quad a\phi_0 = 0, \quad a^*\phi_n = \sqrt{x_{n+1}}\phi_{n+1},
  \quad n=0,1,2, \ldots \; .
\label{shift-ops}
\end{equation}
Then, if $\sum_{n=0}^\infty \dfrac 1{\sqrt{x_n}} = \infty$, the operator $Q = \dfrac {a + a^*}{\sqrt{2}}$ is essentially self-adjoint and hence has a unique self-adjoint extension \cite{borzov,odhorter}, which we again denote by $Q$. It acts on the basis vectors $\phi_n$ in the manner,
\be
  Q\phi_n = \sqrt{\frac {x_n}2}\phi_{n-1} + \sqrt{\frac {x_{n +1}}2}\phi_{n+1}.
\label{recur-reln1}
\end{equation}
From the general theory of self-adjoint operators, there is a Hilbert space $L^2 (\mathbb R , dw)$ (where $dw$ is an even measure) on which $Q$ acts as the operator of multiplication and the $\phi_n$ are functions in this space. Transforming to this space we may then rewrite the above recurrence relation as:
\be
  x\phi_n (x)= \sqrt{\frac {x_n}2}\phi_{n-1}(x) + \sqrt{\frac {x_{n +1}}2}\phi_{n+1} (x),
  \qquad x \in \mathbb R\; ,
\label{recur-reln2}
\end{equation}
which is a set of two-term recurrence relations for a family of orthogonal polynomials, $\phi_n (x)$.
The polynomials can be computed successively, assuming the initial conditions, $\phi_{-1} =0, \; \phi_0 \equiv 1$. The measure $dw$ comes from the spectral family of projectors, $E_x, \; x \in \mathbb R$, of the operator $Q$, in the manner, $dw (x) = \langle \phi_0 \mid E_x \; \phi_0\rangle$.

In general, the measure $dw$ is different from $d\mu$ and so are the sets of polynomials $P_n (x)$ in (\ref{hank-det-pol}) and $\phi(x)$ above. For example, in the case of the canonical coherent states in Section \ref{subsubsec-CCS}, the polynomials $\phi_n(x)$ are the well-known Hermite polynomials, while for the $SU(1,1)$ discrete series coherent states they are the  Pollaczek polynomials mentioned at the end of Section \ref{subsec-pollaczek}, for $\beta = 2j - 2$. For  the Barut-Girardello CS in Section 1.1.3, $x_n = n(2j+n -1)$, however again the corresponding polynomials, which can be computed using (\ref{recur-reln2}), are not well-known.

 In all the above examples, the sequence $x_n$ is strictly increasing and its limit is either a positive number or infinity.   Dickinsen, Pollack and Wannier \cite{Dic:Pol:Wan}  studied a class of orthogonal polynomials of the form in \eqref{recur-reln2} but with $x_n$ replaced by $x_{n+\nu}$ and   requiring  the condition $\{x_n\} \to 0$. Thus, our condition (see (\ref{seq-cond}) below)  is  the exact opposite of  their condition.  The polynomials in the Dickinsen-Pollack-Wannier  class are orthogonal with respect to a discrete measure supported on a countable set  whose only limit point is $x =0$. Later Goldberg \cite{Gol} corrected the claim in \cite{Dic:Pol:Wan} that $x=0$ does not support a mass by showing that $x=0$ may actually  support a positive mass of the orthogonality measure.

   There is a simple way to compute the monic versions of the polynomials $\phi_n(x)$. To se this, note first that in virtue of
(\ref{recur-reln1}) , the operator $Q$ can be represented in the $\phi_n$
basis as the infinite tri-diagonal matrix,
\be
Q = \begin{pmatrix} 0 & b_1 & 0 & 0 & 0 &\ldots \\ b_1 & 0 & b_2 & 0 & 0 & \ldots\\
0 & b_2 & 0 & b_3 & 0 & \ldots\\
0 & 0 & b_3 & 0 & b_4 & \ldots\\0 & 0 & 0 & b_4 & 0 & \ldots\\
\vdots & \vdots &\vdots &\vdots &\vdots &\ddots \end{pmatrix}, \qquad b_n = \sqrt{\frac {x_n}2}\; .
\label{inf-Q-matrix}
\end{equation}
Let $Q_n$ be the truncated matrix consisting of the first $n$ rows and columns of $Q$  and $\mathbb I_n$ the $n\times n$ identity matrix. Then,
\be
  x \mathbb I_n - Q_n =
\begin{pmatrix} x &  - b_1 & 0 & 0 & 0 &\ldots & 0 & 0 & 0\\
- b_1 & x & - b_2 & 0 & 0 & \ldots & 0 & 0 & 0\\
0 & - b_2 & x & - b_3 & 0 & \ldots & 0 & 0 & 0\\
0 & 0 & - b_3 & x & - b_4 & \ldots & 0 & 0 & 0 \\
0 & 0 & 0 & - b_4 & x & \ldots & 0 &0 & 0\\
\vdots & \vdots &\vdots &\vdots &\vdots &\ddots & \vdots & \vdots & \vdots\\
0 & 0 & 0 & 0 & 0 & \ldots & x & -b_{n-2} & 0\\
0 & 0 & 0 & 0 & 0 & \ldots & -b_{n-2} & x & -b_{n-1}\\
0 & 0 & 0 & 0 & 0 & \ldots & 0 & - b_{n-1} & x \end{pmatrix}\; .
\label{trunc-Q-matrix}
\end{equation}
It now follows that the monic polynomial $q_n$, associated to $\phi_n (x)$  is just the characteristic polynomial of $Q_n$ :
\be
  q_n (x ) = \text{det} [ x\mathbb I_n - Q_n ]\; .
\label{monic-polyn}
\end{equation}
These polynomials are related to the $\phi_n (x)$ via
$$
   q_n (x) = \left[ \frac {x!}{2^n}\right]^{\frac 12}\phi_n (x)\; , $$
and satisfy the recurrence relations,
\be
  q_{n+1} (x) = x q_n (x) - \frac {x_n}2 q_{n-1}(x)\; .
\label{mon-recur}
\end{equation}

\section{Generalities about orthogonal polynomials}
To make the paper self-contained, we collect here a few preliminary results on orthogonal polynomials.

\subsection{Recurrence relations and orthogonality measures}\label{subsec-rec-rel-orthog}
Every sequence of monic  polynomials $\{P_n(x)\}$  orthogonal with respect to a positive  measure satisfies a three term recurrence relation of the form \cite{Ism}
\bea
\label{eqgeneral3trr}
\qquad \quad xP_n(x) = P_{n+1}(x) + \al_n P_n(x) + \bt_nP_{n-1}(x), n \ge 0,  \;\; P_0(x) = 1, \bt_0P_{-1}(x):=0.
\eea
To avoid additional assumptions  we only consider polynomials having an orthogonality measure supported on an infinite set.

We start with an even positive  measure $\mu$ supported on $[-L, L]$, $L \le \infty$,   and whose moments
are $\mu_n, n =0, 1, \cdots$.  Further, normalize $\mu$ by $\mu_0 = 1$.  Thus $\mu_{2n+1} =0$ for all $n$ and $\mu_{2n}>0$.  We now assume that $\mu_{2n}$ factors as follows
\bea
\mu_{2n} =2  \int_0^L  t^{2n} d\mu(t) = x_1 x_2 \dots x_n = x_n!, \quad n >0.
\label{eqdefxn}
\eea
(Note that this is always possible by writing $x_n =\mu_{2n}/\mu_{2n-2}$.)  However, we shall only work with cases where $x_n$ has convenient forms, such as being expressible as rational functions of the integer variable $n$.

Note that $2 \mu$ has total mass $=1$ on $[0,L]$.  The representation \eqref{eqdefxn}  clearly implies
\bea
x_n > 0, \quad \textup{for}\; n >0.
\label{seq-cond}
\eea
From here we introduce a family of orthogonal polynomials, of the type (\ref{recur-reln2}) and generated by
\bea
\f_0(x) :=1, \quad \f_1(x) =  \sqrt{\frac 2{x_1}}\; x, \quad x\f_n(x) = \sqrt{\frac {x_{n+1}}2} \f_{n+1} + \sqrt{\frac {x_n}2} \f_{n-1}.
\label{eq3trr}
\eea
We shall mainly explore in this paper the consequences of expressing the moments $\mu$ in the form (\ref{eqdefxn}) and the resulting polynomials of the above type.

\subsection{Additional background material}\label{subsec-addl-bckgr}

It is easy to see that the quadratic form $\sum_{j,k=0}^n \mu_{j+k} y_j \overline{y_k}$ is positive definite, see \cite{Ism}  for example, hence all the Hankel determinants $D_n$,
\begin{equation}
\label{eq2.1.4}
D_n :=\begin{vmatrix}
\mu_0 & \mu_1 & \dotsm & \mu_n \\
\mu_1 & \mu_2 & \dotsm & \mu_{n+1}\\
\vdots & \vdots & \dotsm & \vdots         \\
\mu_n & \mu_{n+1} & \dotsm & \mu_{2n}
\end{vmatrix}.
\end{equation}
 are positive, by the Sylvester criterion.

 The following theorem is important in recovering the absolutely continuous component of an orthogonality measure from the large degree behavior of the orthonormal polynomials.

 \begin{thm}\label{Nevai}
Assume that $p_n$ satisfies \eqref{eqgeneral3trr} with $\al_{n-1} \in \R$ and $\bt_n >0$ for all $n >0$. In addition assume that
\begin{equation}
\label{eq11.2.3}
\Sum\left[\left|\sqrt{\beta_n}-\frac12\right|+\left|\alpha_n\right|\right] < \infty,
\end{equation}
then $\mu$ has an absolutely continuous component $\mu'$ supported on $[-1,1]$. Furthermore if $\mu$ has a discrete part,
then it will lie outside $(-1,1)$. In addition the limiting relation
\begin{equation}
\label{eq11.2.4}
\limsup_{n\to\infty}\left[\sqrt{1-x^2}\,\frac{P_n(x)}{\sqrt{\zeta_n}}
-\sqrt{\frac{2\sqrt{1-x^2}}{\pi\mu'(x)}}\sin((n+1)\theta-\varphi(\theta))\right]=0,
\end{equation}
holds, with $x=\cos\theta\in(-1,1)$. In \eqref{eq11.2.4} $\varphi(\theta)$ does not depend on $n$. Here $\zeta_n = \bt_1\bt_2\cdots \bt_n$.
\end{thm}

This theorem is due to Nevai \cite{Nev}.
Observe that Nevai's theorem then relates the asymptotics of general polynomials to those of the Chebyshev polynomials.  Note also that Nevai's theorem  gives a scattering amplitude and says that there may be  a phase shift.

\begin{thm}\label{IsmLi}
Let $p_n(x)$ be generated by \eqref{eqgeneral3trr}. Then the zeros of the polynomial $p_n(x)$  lie in
$(A,B)$, where
\bea
\notag
B = \textup{max}\{\xi_j: 0 < j < n\}, \qquad  A =  \textup{min} \{\eta_j: 0 < j < n\},
\eea
where $\eta_j \le \xi_j$ and
\bea
\label{eqdefxiyi}
\xi_j, \eta_j = \frac{1}{2} (\al_j+ \al_{j-1})  \pm \frac{1}{2}\;  \sqrt{(\al_j-\al_{j-1})^2 + 16 \bt_j}, \quad 1 \le j < n.
\eea
\end{thm}
Theorem \ref{IsmLi} is the  special case $c_n =1/4$  of a result due to Ismail and Li in \cite{Ism:Li}. The full result is also stated and proved in  \cite[Theorem  7.2.7]{Ism}.

The zeros of orthogonal polynomials are real and simple, so we shall follow the standard notation
in   \cite{Ism} or \cite{Sze} and arrange the zeros $x_{n,j}, 1 \le j \le n$ as
\bea
\label{eqorderzeros}
x_{n,1} > x_{n,2} > \cdots > x_{n,n}.
\eea

Berg and Duran \cite{Ber:Dur} proved the following theorem.
\begin{thm} \label{BergDuran}
Assume that $\{1/x_n\}$ is a Hausdorf moment sequence, that is there is a probability measure $\nu$ supported on $[0,1]$ such that $1/x_n  = \int_0^1 t^n d\nu(x)$. Then there is a probability measure $\xi$ upported on $[0, \infty)$ such that $x_1x_2 \cdots x_n = \int_0^\infty t^n d\xi(x)$, that is
$\{x_1x_2 \cdots x_n\}$ is a Stieltjes moment sequence.
\end{thm}

Many examples and applications are in \cite{Ber1}, \cite{Ber2}.
In the theory of coherent states we start with the sequence $\{x_n\}$ and construct the measure $\la$, so Theorem 2.3  gives a sufficient condition for the existence of the measure $\la$.

\section{Results and examples}\label{sec:res-ex}
We enunciate a few results in this section and work out some examples.

\subsection{Bounds}\label{subsec:bounds}
 Our first result is the monotonicity of $\{x_n\}$.
 \begin{thm} \label{bounded}
 The sequence $\{x_n:n =1, 2, \cdots\}$ is strictly increasing. If $L< \infty$ then $x_n < L^2$.
 \end{thm}
 \begin{proof}
 Clearly
\begin{equation}
0 < D_2 = \begin{vmatrix}
1 & 0 &  x_1 \\
0 & x_1 &  0 \\
x_1 & 0 & x_1 x_{2}
\end{vmatrix} = x_1^2(x_2-x_1).
\end{equation}
Therefore $x_2 > x_1$.
We note that $x^{2n} d\mu$ is a positive even measure supported on $ [-L,L]$ and its $D_2$ is
\begin{equation}
\begin{vmatrix}
\mu_{2n} & 0 &  \mu_{2n+2} \\
0 & \mu_{2n+2} &  0 \\
\mu_{2n+2} & 0 & \mu_{2n+4}
\end{vmatrix} = \mu_{2n+2} [\mu_{2n}]^2\, [x_{n+1}x_{n+2} -x_{n+1}^2],
\end{equation}
which implies $x_{n+2} > x_{n+1}$, for $n =0, 1, \cdots$.  The same conclusion also follows from
\bea
\notag
\begin{gathered}
0 < 2\int_0^L t^{2n} (t^2-x_{n+1})^2 d\mu(t) = \mu_{2n+4} - 2x_{n+1}\mu_{2n+2} + x_{n+1}^2 \mu_{2n}\\
= \mu_{2n+2} \left[ x_{n+2}-2 x_{n+1} +x_{n+1}    \right].
\end{gathered}
\eea
If $L$ is finite we use  the definition of
$\mu_{2n} = x_1 \cdots x_n$ and conclude that
$$
x_1 \cdots x_nL^2 -  x_1 \cdots x_n x_{n+1} = 2\int_0^L t^{2n} [L^2 -t^2] d\mu(x) >0.
$$
This implies $x_{n+1} <  L^2$ for all $n \ge 0$.
 \end{proof}

 A partial converse to the above theorem is the following.
 \begin{thm} \label{unbounded}
If  $L = \infty$ then the  sequence  $\{x_n\}$ is unbounded.
\end{thm}
\begin{proof}
Assume $L = \infty$ and $x_n \le M$, and $M>1$. Then for every $A >0$ the integral $\int_{[A, \infty)} d\mu(t) >0$.  It is clear that
\bea
\notag
\begin{gathered}
M^n > x_1 x_2 \cdots x_n = 2 \int_{[0,\infty)}  t^{2n} d\mu(t)   \ge  2\int_{[2M, \infty)}t^{2n} d\mu(t)  \\
\ge 2 [2M]^{2n}\int_{[M, \infty)} d\mu(t).
\end{gathered}
\eea
Therefore $M > 4M^2 \left[2\int_{[2M, \infty)} d\mu(t)\right]^{1/n}$ which is impossible for sufficiently large $n$.
\end{proof}

Note that if $L\le 1$ then $x_n < 1$ and $1/x_n$ is never a Hausdorff moment sequence. In these cases the
assumptions of the  Berg--Durand theorem, Theorem \ref{BergDuran},  are not satisfied.

\bigskip

 \noindent{\bf Remark}:
It is not clear that we can say much about the case $L = \infty$ so we will assume that $L$ is bounded throughout  the rest of this section. Therefore
$\{x_n\}$ is a monotone sequence converging to $M$, say.
One can derive nonlinear inequalities satisfied by the $x_n$'s.  For example using
$$
0 < 2\int_0^L t^{2n}(t^2-x_{n+1})^2(t^2-x_{n+2})^2 d\mu(t)
$$
one gets
\bea
\begin{gathered}
x_{n+3}(x_{n+4}-x_{n+2}) + x_{n+1} (x_{n+2}-x_{n+1}) \\ > x_{n+2}(x_{n+3}-x_{n+2}) +
2 x_{n+1}(x_{n+3}-x_{n+2}).
\end{gathered}
\label{eqineq1}
\eea
Of course one can integrate other factions like
$t^{2n}(t^2-x_{n+1})^4$ or $t^{2n} (t^2-x_{n+1})^2(t^2-x_{n+3})^2$
and get other inequalities. On the other hand  by expanding  the determinant Hankel determinant
\begin{equation}
\notag
\begin{vmatrix}
\mu_{2n} & 0 & \mu_{2n+2} & 0 &  \mu_{2n+4} \\
0  & \mu_{2n+2} & 0  & \mu_{2n+4} & 0 \\
 \mu_{2n+2} & 0 & \mu_{2n+4} & 0 & \mu_{2n+6}         \\
0 &  \mu_{2n+4} & 0 & \mu_{2n+6}  & 0\\
\mu_{2n+4} & 0 & \mu_{2n+6}  & 0 & \mu_{2n+8}
\end{vmatrix} > 0,
\end{equation}
and after deleting the positive terms we obtain the necessary condition
\bea
\label{eqineq2}
2x_{n+1}x_{n+2}x_{n+3} + x_{n+2}x_{n+3} x_{n+4} >   x_{n+1} x_{n+1}^2+ x_{n+2} x_{n+3}^2
+ x_{n+1}x_{n+3}x_{n+4}.
\eea

 \begin{thm}
 Let $L$ be finite and let $
 M := \lim_{n\to \infty} x_n$.  If
$$
\sum_{n=1}^\infty |\sqrt{x_n} - \sqrt{M}| \;  \textup{converges}.
$$
 Then the  orthogonality measure of $\f_n(x)$ has an absolutely  continuous component supported on the interval
 $[-2\sqrt{M}, 2\sqrt{M}]$. Moreover all the zeros of the polynomials lie in $(-2\sqrt{M}, 2 \sqrt{M})$, hence  the  discrete part of the orthogonality measure  is either empty or has  two discrete masses (bound states) at
$x = \pm 2\sqrt{M}$.
\end{thm}
\begin{proof}
Let   $x =y/2\sqrt{M}$  and put $\f_n(x) = \psi_n(y)$ and apply Theorem \ref{Nevai} to see that $\psi_n$'s are orthogonal with respect to a measure whose absolutely continuous component is supported on $[-1,1]$. Thus the orthogonality measure of $\f_n$ has an absolutely continuous component supported on $[-2\sqrt{M}, 2\sqrt{M}]$. Next apply Theorem \ref{IsmLi}. In the present case $\al_n =0, \bt_n =x_n$, hence
$\xi_j, \eta_j = \pm 2 \sqrt{x_j}$. The monotonicity of the $x_n$s shows that the zeros of $\f_n(x)$ belong to
($-\sqrt{x_n}, \sqrt{x_n})$. From Theorem \ref{Nevai} we conclude that the discrete part of the orthogonality measure is outside $(-2\sqrt{M}, 2\sqrt{M})$. If $[A, B]$ is the smallest interval containing the support of the orthogonality measure then  the largest and smallest zeros of $\f_n$ converge to $B$ and $A$, respectively. Thus $[A, B] = [-2\sqrt{M}, 2\sqrt{M}]$. This shows that the discrete part may only occur at $\pm  2\sqrt{M}$.
\end{proof}

\subsection{Example 1}\label{subsec-ultraspher} Consider the ultraspherical polynomials $\{C_n^\nu(x)\}$ where
$$
d\mu(x) = \frac{\Gamma(\nu+1)(1-x^2)^{\nu-1/2}}  {\sqrt{\pi}\; \Gamma(\nu+1/2)} \, dx,  \quad x \in [-1,1], \nu > -1/2.
$$
Now
\bea
\begin{gathered}
x_1x_2 \cdots x_n =  \frac{2\Gamma(\nu+1)}{\sqrt{\pi}\; \Gamma(\nu+1/2)}
 \; \int_0^1 x^{2n}(1-x^2)^{\nu-1/2} dx  \\
 = \frac{\Gamma(\nu+1)\Gamma(n+1/2)}{\sqrt{\pi}\; \Gamma(\nu+n+1)} = \frac{(1/2)_n}{(\nu+1)_n}
\end{gathered}
\eea
Therefore
$$
x_n = \frac{n-1/2}{\nu+n}.
$$
Here is an interesting point to show how sharp the monotonicity of the $x_n$'s is. An easy calculation is to show that $x_n < x_{n+1}$ is equivalent to the integrability of the weight function, namely $\nu > -1/2$.  The monic recurrence relation for the family of polynomials in \eqref{eq3trr}, after replacing $x$ by $2x$,   is
$$
2x u_n(x) =  u_{n+1}(x) + \frac{n-1/2}{\nu+n} \, u_{n-1}(x),
$$
which is not a standard polynomial. It is a special case of the associated Pollaczek polynomials, see
\cite[Chapter 5]{Ism}.  We know the absolutely continuous component of its orthogonality measure but we do not know whether $x = \pm 1$ support any discrete  masses.

\subsection{Example 2}\label{subsec-pollaczek}
This is more general than Example 1. Consider the absolutely continuous  measure
$$
d\mu(x; \al, \bt) = \frac{\Gamma(\al+\bt+3/2)|x|^{2\al}(1-x^2)^{\bt}}  {\Gamma(\al+1/2)\; \Gamma(\bt+1)} \, dx,
\quad x \in [-1,1],
$$
where $\al > -1/2, \, \bt > -1$. This is essentially the measure $\la$ in  \eqref{mom-prob}.  The polynomials in this case are defined according to their parity. The polynomials of even degree are constant  multiples of the Jacobi polynomials $P_n^{(\al-1/2, \beta)}(1-2x^2)$ while the odd degree ones are constant  multiples of the Jacobi polynomials $xP_n^{(\al+1/2, \beta)}(1-2x^2)$.

We then have
\bea
\begin{gathered}
x_1x_2 \cdots x_n = \frac{\Gamma(\al+\bt+3/2)}  {\Gamma(\al+1/2)\; \Gamma(\bt+1)}  \; \int_0^1 x^{2n+2\al}(1-x^2)^{\bt} dx  \\
=    \frac{(\al+1/2)_n}{(\al+\bt+3/2)_n}.
\end{gathered}
\eea
This gives
$$
x_n = \frac{\al+n-1/2}{\al+\bt+n+1/2}.
$$
The case $\al =1/2$ gives the Pollaczek polynomials with parameters $\lambda = (\bt+1)/2, a = (\bt+1)/2, b =0$, see \S 5.4 and \S 5.5 of \cite{Ism}.
 If $\al \ne 1/2$ we get the  associated Pollaczek polynomials.  They are given at the end of \cite{Cha:Ism}.

 For completeness It may be of interest to say some thing about the Pollaczek polynomials. They can be defined   by  the recurrence relation,   \cite{Chi}, \cite{Ism}
\begin{equation}
\label{Pollrr}
\begin{gathered}
(n+1) P_{n+1}^\lambda(x;a,b)=2[(n+\lambda+a)x+b]P_n^\lambda(x;a,b) \\
\mbox{} \qquad
-(n+2\lambda-1) P_{n-1}^\lambda(x;a,b), \quad n >0,
\end{gathered}
\end{equation}
and the initial conditions
\begin{equation}
\label{Pollic}
P_0^\lambda(x;a,b)=1, \quad P_1^\lambda(x;a,b)=2(\lambda+a)x+2b.
\end{equation}
Their hypergeometric representation, orthogonality relation  and generating  functions can be found in \S 5.3 of \cite{Ism}. The orthogonality restricts $\la. a, b$ to be in a certain subset of $\R^3$. The parameter  domain is further divided into subsets according to the nature of their orthogonality measure. The measure always has  an  absolutely continuous component supported on $[-1,1]$. In addition it may have an empty,  finite or infinite discrete part  depending on where $\la,  a, b$ lies in the parameter domain. This is described in detail in \cite{Cha:Ism}, see also \S 5.3 in \cite{Ism}.  The Pollaczek polynomials also appeared in the $J$ matrix method for discretization of the continuum where the energy parameter $E$ for the hydrogen atom is related to $x$ in the Pollaczek polynomials  via $x = (E-1/8)/(E+1/8$. The details are in \cite{Yam:Rei}, \cite{Hel:Rei:Yam}, see also \cite[\S 5.8]{Ism}. The latter reference records the explicit form of the measure in different parts of the parameter domain.  It is interesting to note that the Pollaczek polynomials also appear in the relativistic Coulomb problem as in the work of Abdulaziz Alhaidari \cite{AlHa-AnnPhys2004} and Charles
Munger \cite{Mung}.

\subsection{Further examples}\label{subsec:furth-ex}
We consider two additional examples of measures, with moments written in the form \eqref{eqdefxn}.
Consider the integral, \cite[(27),p.51]{Erd:Mag:Obe:Tri2},
\bea
\label{eqKnuint1}
\int_0^\infty K_{2\nu}(\bt t)t^{2\mu-1}dt = 2^{2\mu-2}\bt^{-2\mu} \Gamma(\mu+\nu)\Gamma(\mu-\nu)),
\eea
$\Re\,  (\mu\pm \nu+ > 0, \Re\,   \bt >0.$ The weight function
\bea
w(x) = \frac{2^{1-2\mu}\bt^{2\mu}}{ \Gamma(\mu+\nu)\Gamma(\mu-\nu))} K_{2\nu}(\bt |x|)|x|^{2\mu-1}, \quad x \in \R.
\eea
Therefore
\bea
\mu_{2n} = 4^n\bt^{-2n}(\mu+\nu)_n(\mu-\nu)_n,
\notag
\eea
and
\bea
x_n = (4/\bt^2)(\mu+\nu+n-1) (\mu-\nu+n-1).
\eea
In this case the polynomials generated by (\ref{eq3trr}) are the associated Meixner-Pollaczek polynomials.  Weight functions for these polynomials have been computed in \cite{Cha:Ism}.

 Consider the integral
\bea
\begin{gathered}
\int_0^\infty e^{-at}K_\nu(\bt t) t^{\mu-1} dt =
\frac{\sqrt{\pi}\; (2\bt)^\nu\,  \Gamma(\mu+\nu) \Gamma(\mu-\nu)}
{\Gamma(\mu+1/2)(a+\bt)^{\mu+\nu}}\\
\times {}_2F_1 \left(\left.\begin{matrix}
\mu+\nu, \nu+1/2 \\
\mu+1/2
\end{matrix}\, \right|\frac{a-\bt}{a+\bt}\right),
\end{gathered}
\eea
valid for $\Re (\mu\pm \nu >0, \Re (a+\bt) >0.$
This is (26), page 50 of \cite{Erd:Mag:Obe:Tri2}.

\bigskip

\noindent {\bf Eample 1}: To sum the ${}_2F_1$ we are forced to take $a = \bt$, hence there is no loss of generality in choosing  $a = \bt =1$. Consider the even normalized weight function
\bea
\label{eqKnuweight}
w(t) := \frac{\Gamma(\mu+1/2)2^{\mu}} {\sqrt{\pi}\; \,  \Gamma(\mu+\nu) \Gamma(\mu-\nu)}    e^{-t^2}K_\nu(t^2) |t|^{2\mu-1}, \quad t \in \R.
\eea
Thus
\bea
\mu_{2n} = \frac{(\mu+\nu)_{n}(\mu-\nu)_{n}}{2^n(\mu+1/2)_{n}},
\notag
\eea
and
we find $x_n$ in \eqref{eqdefxn} is
\bea
\label{eqNew1xn}
x_n = \frac{ (\mu+\nu +n-1) (\mu-\nu +n-1)}{2(\mu+n-1/2)}.
\eea
The $x_n$s are unbounded as expected. Nothing is known about the polynomials
 generated by \eqref{eq3trr} with $x_n$ defined by \eqref{eqNew1xn}.

\bigskip
Let
\bea
w(t) := \frac{\Gamma(\mu+1/2)2^{\mu-1}} {\sqrt{\pi}\; \,  \Gamma(\mu+\nu) \Gamma(\mu-\nu)}    e^{-|t|}K_\nu(|t|) |t|^{\mu-1}, \quad t \in \R.
\eea
Thus
\bea
\mu_{2n} = \frac{(\mu+\nu)_{2n}(\mu-\nu)_{2n}}{4^n(\mu+1/2)_{2n}},
\notag
\eea
and
we find for $x_n$ in \eqref{eqdefxn}
\bea
\qquad x_n = \frac{(\mu+\nu+2n-2)(\mu+\nu +2n-1)(\mu-\nu+2n-2)(\mu-\nu +2n-1)}{4(\mu+2n-3/2)(\mu+2n-1/2)}.
\eea
At  first glance these polynomials seem to be symmetric continuous Hahn polynomials,
 \cite[(1.4.2)]{Koe:Swa} with $a = c$ and $b =d$. A closer examination however shows that  this is not the case and the polynomials generated by \eqref{eq3trr} with the above $x_n$'s are new.

\subsection{Completely Monotonic Functions}\label{subsec:comp-mon-fcns}
Bustoz and Ismail \cite{Bus:Ism2} proved that the function
\bea
f(x;a,b):= \frac{\Gamma(x) \Gamma(x+a +b)}{\Gamma(x+a)\Gamma(x+b)}, \quad a, b \ge 0,
\eea
is completely monotonic, that is $(-1)^n \frac{d^n}{dx^n}f(x, a, b) \ge 0$ on $(0, \infty)$.  Therefore the function
\bea
g(x; a, c): = \frac{\Gamma(a)\Gamma(b)}{\Gamma(c) \Gamma(a +b-c)}
 \frac{\Gamma(x+c) \Gamma(x+a +b-c)}{\Gamma(x+a)\Gamma(x+b)},
\eea
is completely monotonic for $a\ge c, b \ge c, c \ge 0.$ When $c >0$, $g$ is completely monotonic on $[0,\infty)$. By Bernstein's theorem \cite{Wid} there is a unique probability measure $\al(x)$ supported on a subset of $[0, \infty)$ such that
\bea
\label{eqgasLaplace}
g(x; a, b, c)= \int_0^\infty e^{-xt} d\al(t).
\eea
In fact Bustoz and Ismail \cite{Bus:Ism1} proved that the corresponding probability distribution is infinitely divisible
\cite{Fel}. Now the measure $\mu(u):= \frac{1}{2} \al(-2\ln |u|)$ is an even probability measure on $\R$ and its $2n$-th moment is
\bea
\notag
\int_\R u^{2n} d\mu(u) = 2\int_0^1 u^{2n} d\mu(u) = g(n;a,b, c) = x_1 x_2 \cdots x_n.
\eea
This gives
\bea
x_n = \frac{(c+n-1)(a+b-c+n-1)}{(a+n-1)(b+n-1)}.
\notag
\eea
The polynomials generated by \eqref{eq3trr} when $c =1, a = \nu+1, b =  \nu$  are the orthonormal ultraspherical polynomials
$$
\sqrt{\frac{n!(n+\nu)}{(2\nu)_n}}C_n^\nu(x/2),
$$
\cite{Rai}, \cite{Erd:Mag:Obe:Tri2}, \cite{Ism}. In general we choose $a = \nu+ 1 +c , b = \nu+c$ and
keep $c$ as an association parameter. The polynomials become constant multiples of orthonormal  associated ultraspherical polynomials at $x/2$,  \cite{Bus:Ism1}.

\bigskip

Let $0 < q < 1$. The $q$-shifted factorials are
\bea
(a;q)_0 =1, (a;q)_n = \prod_{k=1}^{n}(1- aq^{k-1}), \quad n = 1, 2, \cdots, \; \textup{or}\; \infty.
\eea
The $q$-Gamma function $\Gamma_q(x)$  is, \cite{And:Ask:Roy}, \cite{Gas:Rah}
\bea
\Gamma_q(x) =(1-q)^{1-x} \prod_{k=0}^\infty \frac{1-q^{k+1}}{1-q^{x+k}}.
\eea
It satisfies the functional equation
\bea
\notag
\Gamma_q(x+1) = \frac{1-q^x}{1-q} \, \Gamma_q(x).
\eea
It also has the initial values $\Gamma_q(1) = \Gamma_q(2) =1$.

We now consider the function
\bea
h(x; a, c): = \frac{\Gamma_q(a)\Gamma_q(b)}{\Gamma_q(c) \Gamma_q(a +b-c)}
 \frac{\Gamma_q(x+c) \Gamma_q(x+a +b-c)}{\Gamma_q(x+a)\Gamma(x+b)}.
\eea
A special case of a result of Ismail and Muldoon \cite{Ism:Mul}  is that $h(x; a, c) = e^{-H(t)}$ and  $H$ is completely monotonic on $[0,\infty)$ for $a\ge c, b \ge c, c \ge 0.$

Now let $\bt$ be the probability measure defined by
\bea
h(x; a, c) =  \int_0^\infty e^{-xt} d\bt(t).
\eea

As in the case of \eqref{eqgasLaplace} we let $\nu(u):= \frac{1}{2} \bt(-2\ln |u|)$ and set
\bea
\notag
\int_\R u^{2n} d\mu(u) = 2\int_0^\infty u^{2n} d\mu(u) = g(n;a,b, c) = x_1 x_2 \cdots x_n.
\eea
With $A = q^a, B = q^b, C = q^c, D = q^d$, we have
\bea
x_n = \frac{(1-Cq^{n-1})(1- ABCq^{n-1})}{(1-Aq^{n-1})(1- Bq^{n-1})}.
\notag
\eea
When $C= q, A/q = B = \beta$ the polynomials are the $q$-ultraspherical polynomials  of Askey and Ismail \cite{Ask:Ism}
\bea
\notag
\sqrt{ \frac{(q;q)_n(1- \bt q^n)}{(\bt^2;q)_n}}\;  C_n(x/2;\bt|q).
\eea
A complete treatment of the  $q$-ultraspherical polynomials  is available in \cite{Ism}. If $C \ne q$ we get the associated $q$-ultraspherical  polynomials of \cite{Bus:Ism1} or associated symmetric $q$-Pollaczek polynomials, \cite{Cha:Ism}.

Many quotients of products of Gamma (repectively  $q$-Gamma) functions  are known to be completely monotonic, see for example \cite{Gri:Ism} and \cite{Ism:Mul}. Each combination gives rise to orthogonal polynomials $\f_n(x)$ of the type generated by \eqref{eq3trr}, where $x_n$ is a quotient of two monic polynomials of $n$ (respectively of $q^n$) of the same degree. Therefore we can always generate many cases which are not in the literature but to which the results of this section apply.    We end this section with few examples of completely monotonic functions and the write down the corresponding sequence $\{x_n\}$.

In order to state the more general results alluded to above we need some additional notation. Let $S_n$ be the
set (group) of all permutations on $n$ symbols, $a_1, a_2, \dots, a_n$. Let $O_n$, and $E_n$ be the
sets of odd, and even permutations over $n$ symbols, respectively.
 Moreover let $P_{n,k}; 1\le k \le n$ be the set of all vectors ${\bf m} = (m_1, m_2, \cdots, m_k)$
  such that $1 \le m_r < m_s \le n$ for $1 \le r < s \le k$; and $P_{n,0}$ is defined as the empty set.

\begin{thm}\textup{(\cite{Gri:Ism})}
Let $a_1 > a_2 > \dots > a_n \ge 0$ and define
\bea
F(x) = \frac{\prod_{\sigma \in E_n} \left[\G(x+ a_{\s(2)}+ 2 a_{\s(3)}
+\dots + (n-1) a_{\s(n)}\right]}{\prod_{\sigma \in O_n} \left[
\G(x+ a_{\s(2)}+ 2 a_{\s(3)} +\dots +
(n-1) a_{\s(n)}) \right]}.
\eea
Then $F(x-a_2-2a_3 -\dots - (n-1)a_n) = e^{-H(x)}$ and $H^\prime$
is completely monotonic, hence $F$ is completely monotonic. The same conclusion holds for the function
\bea
F(x, q) = \frac{\prod_{\sigma \in E_n} \left[\G_q(x+ a_{\s(2)}+ 2 a_{\s(3)}
+\dots + (n-1) a_{\s(n)}\right]}{\prod_{\sigma \in O_n} \left[
\G_q(x+ a_{\s(2)}+ 2 a_{\s(3)} +\dots +
(n-1) a_{\s(n)}) \right]}.
\eea
\end{thm}

Note that
$$\sum_{k=1}^n (k-1)(a_{\s(k)}-a_k) \ge 0$$
 holds for any permutation $\s$ when $a_1 > a_2> \cdots >a_n >0$.
 \begin{thm}\label{6.3}
 The function
 \bea
 \label{eq6.3}
 F_s(x) =
 \frac{\G(x)
  \prod_{k=1}^{\lfloor{s/2}\rfloor}\left[\prod_{{\bf m}\in {P_{s,2k}}}\G(x+\sum_{j=1}^{2k}a_{m_j}\right]}
  { \prod_{k=1}^{\lfloor{(s+1)/2}\rfloor}\left[\prod_{{\bf m}\in {P_{s,2k-1}}}\G(x+\sum_{j=1}^{2k-1}a_{m_j}\right]},
  \eea
  is of the form $e^{-H(x)}$ and $H^\prime(x)$ is completely monotonic, hence $F_n$ is completely monotonic.
 \end{thm}
Theorem \ref{6.3} is due to Grinshpan and Ismail  \cite{Gri:Ism}.  In particular the case $n =3$, after shifting $x$ by $a_0$,  says that  the function
\bea
\label{eq6.4}
\frac{\G(x+a_0)\G(x+a_0+ a_1+a_2)\G(x+a_0+a_1+a_3)\G(x+a_0+a_2+ a_3)}{\G(x+a_0+a_1)\G(x+a_0+a_2)\G(x+a_0+ a_3)\G(x+a_0+a_1+a_2+a_3)}
\eea
is completely monotonic for $a_1\ge a_2\ge a_3 \ge 0$, and $a_0 \ge 0$. Choosing $a_0=1$ and dividing the above function by its value at $x = a_0 =1$  we find that  the corresponding $x_n$'s are given by
\bea
x_n = \frac{n(n + a_1+a_2)(n+a_1+a_3)(n+a_2+ a_3)}{(n+a_1)(n+a_2)(n+ a_3)
(n +a_1+a_2+a_3)}.
\eea
 In this case $x_n \to 1$ and $|\sqrt{x_n}- 1| = \O(1/n^2)$, as $n \to \infty$,  hence $\sum_{n=1}^\infty |\sqrt{x_n}- 1| < \infty$ and the conclusions of Nevai's theorem hold.

 \begin{rem}
 Let $\mu$ be a probability measure such  that
 $$
 F_s(x+a_0)/F_s(a_0) = \int_0^\infty e^{-xt} d\mu(t) = \int_0^1 u^{2x} d\nu(u)
 $$
 where we performed the change of variables $e^{-t/2} =u$. Therefore
 $$
F_s(n+a_0)/F_s(a_0) = \int_0^1 u^{2n} d\nu(u)  = \int_{-1}^1 u^{2n}\left[\frac{1}{2} d\nu(|u|)\right],
 $$
 where $\nu$ is extended as an even measure, so that $\frac{1}{2}d\nu$ is now a probability measure on $[-1,1]$.  It is clear we can define $x_n$ by
 \bea
 \label{eqdefxnFs}
 x_n = F_s(n+a_0)/F_s(n+a_0-1), n \ge 1.
 \eea
 We will show below that the $x_n$s defined this way have the property $|\sqrt{x_n}- 1| = \O(1/n^2)$, $n \to \infty$, so Nevai's theorem is applicable.
 \end{rem}
We now show that the $x_n$'s defined by \eqref{eqdefxnFs} satisfy the conditions in Nevai's theorem. Observe that the number of terms in the numerator in \eqref{eq6.3} is
\bea
\notag
1+ \binom{s}{2} + \binom{s}{4} + \cdots + \binom{s}{2\lfloor{s/2}\rfloor} = 2^{s-1},
\eea
while the number of terms in the denominator in \eqref{eq6.3} is
\bea
\notag
\binom{s}{1}+ \binom{s}{3} + \binom{s}{5} + \cdots + \binom{s}{2\lfloor{(s+1)/2}\rfloor} = 2^{s-1}.
\eea
So the number of terms in the numerator and denominator in  \eqref{eq6.3}  is the same. Now the sum of the arguments of the Gamma functions in the numerator and denominator in  \eqref{eq6.3} is
\bea
\notag
\begin{gathered}
2^{s-1}+ 2\binom{s}{2} + 4 \binom{s}{4} + \cdots + 2\lfloor{s/2}\rfloor \binom{s}{2\lfloor{s/2}\rfloor},  \quad \textup{and}\\
2^{s-1} + \binom{s}{1}+ 3 \binom{s}{3} +5  \binom{s}{5} + \cdots +2\lfloor{(s+1)/2}\rfloor  \binom{s}{2\lfloor{(s+1)/2}\rfloor}.
\end{gathered}
\eea
respectively. But both are equal to $2^{s-1} + s 2^{s-2}$. After discarding the $x's$ each sum of the remaining terms is  $s 2^{s-2}$. But both  the numerator and denominator  sums  are  symmetric functions of the $a_j$s hence each $a_j$ appears   $2^{s-2}$ times. This  and the monotonicity  of $F_s$ show that
 $1 > x_n = 1 -\O(1/n^2)$.

\bigskip

It would be of great interest to study the coherent states arising from the above weight functions and associated sequences $\{x_n\}$. Many of these are expected to have physical significance.

\section*{Acknowledgements}
The research of Mourad E.H. Ismail is supported by NPST Program of King Saud University; project number 10-MAT1293-02 and  by 
 Research Grants Council of Hong Kong under  contract \#  101410. The work of S.T. Ali was partially supported by a grant from the Natural Science and Engineering Research Council (NSERC) of Canada. Also part of this work was carried out while STA was visiting MEHI at the City University of Hong Kong. He would like to thank the Department of Mathematics, City University of Hong Kong, for hospitality.


\noindent Email: M.E.H.I. : mourad.eh.ismail@gmail.com, S.T.A. : stali@mathstat.concordia.ca

\end{document}